%% file: main.tex
\documentclass[12pt, draftclsnofoot, onecolumn]{IEEEtran}

\usepackage[english]{babel}
\usepackage[utf8]{inputenc}
\usepackage{amsfonts,dsfont,amsthm,amsmath,amssymb,amscd,epsfig,
			bm,subfigure,graphicx,mathtools}
\usepackage{url}
\usepackage{booktabs}
\usepackage[noadjust]{cite}

\usepackage{tikz}
\usetikzlibrary{decorations.markings}
\usepackage{pgfplots, pgfplotstable}
\pgfplotsset{compat=1.10}
\usepgfplotslibrary{fillbetween}
\usetikzlibrary{arrows,calc}

\newtheorem{proposition}{Proposition}
\newtheorem{theorem}{Theorem}
\newtheorem{corollary}{Corollary}
\newtheorem{lemma}{Lemma}

\newcommand{\euler}{\mathrm{e}}

\def\N{{\mathbb N}}
\def\R{{\mathbb R}}

\def\P{{\mathbb P}}

\def\E{{\mathbb E}}

\def\P{{\mathbb P}}
\def\ind{\mathds{1}}
\def\cal{\mathcal}

\begin{document}

\author{Plínio~S.~Dester,
		Paulo~Cardieri,
        José~M.~C.~Brito%
        }

\title{On the Stability of a $N$-class Aloha Network}

\date{\today}

\IEEEtitleabstractindextext{%
\begin{abstract}
        Necessary and sufficient conditions are established for the stability of a high-mobility $N$-class Aloha network, where the position of the sources follows a Poisson point process, each source has an infinity capacity buffer, packets arrive according to a Bernoulli distribution and the link distance between source and destination follows a Rayleigh distribution. It is also derived simple formulas for the stationary packet success probability and mean delay.
\end{abstract}}

\maketitle
\IEEEdisplaynontitleabstractindextext

\section{Introduction}

Since the first deployments of large-scale wireless communications systems based on cellular technology in the mid-1970s, there has been an increasing demand for wireless communication services, which has led to the permanent search for more efficient use of radio resources. This situation is now more exacerbated, with applications that require higher data rates, such as those based on video streaming services, or scenarios with a larger number of terminals, as in the situations envisaged by the Internet of Things. In this sense, next-generation system developers and service providers are facing perhaps unthinkable challenges in the 1970s. Challenges, such as data rates of up to tens of Gb/s, latency in the order of milliseconds, and reduced energy consumption to 10\% of current consumption, are set as targets for the fifth-generation cellular system (5G System) \cite{Boccardi2014}. 

In one scenario envisioned for the 5G systems, a number of subnetworks will co-exist in the same geographic area, sharing radio resources. Each of these subnetworks will be dedicated to serve a particular type of application and/or scenario, with its own requirements, such as coverage, transmission rates and maximum acceptable latency \cite{Ghosh2012}. In fact, it seems to be a consensus in the academic and industrial communities that the goals imposed on 5G systems will only be achieved through the use of heterogeneous networks. Thus, the evaluation of the performance of such systems and the design of techniques that efficiently exploit radio resources require a better understanding of the mechanisms involved in the transmission of a message through the wireless medium.

In this work, we are interested in studying the performance of a heterogeneous network in which $N$ classes of users share the same radio resources, \emph{i.e.}, radio spectrum and transmission power. Each of these user classes has its own characteristics, such as terminal density, transmit power and traffic intensity, and quality of service requirements, namely, communication link quality and maximum tolerable delay. Scenarios like this one are expected to be found in 5G systems, involving, for example, applications of Internet of Things, in which thousands of wireless terminals connected to sensors access the wireless network to transmit their messages. These wireless connections may involve an access point, in a cellular mode, or, alternatively, terminals may communicate directly with each other, in the so-called device-to-device mode (D2D). Terminals may be associated with different applications, with different quality of service requirements, such as maximum acceptable delay and minimum transmission rate.

The scenario studied in this paper has been investigated in several studies found in the literature. A particular interest has been observed in the situation where packets arrive at the terminals randomly, and packets waiting for transmission are stored in queues. In such a situation, mutual interference among terminals makes the queues of the terminals coupled, since the transmission success probability of a terminal (\emph{i.e.}, the service rate of the queue associated with that terminal) depends on the state of the queues of other terminals (if their queues are either empty or non-empty). The analysis of networks with coupled queues is known to be difficult, especially when the capture model\footnote{According to the capture model, a packet is successfully received if the corresponding signal to interference plus noise ratio (SINR) at the receiver is above a certain threshold. In contrast, the collision model states that a packet transmission is successful only if there are no concurrent transmissions.} is adopted. To overcome this difficulty, several authors have used the concept of stochastic dominance (see, for instance,  \cite{Ephremides1988}, \cite{Naware2005}), which allows to determine the conditions for queue stability.

Stamatiou and Haenggi \cite{stamatiou2010random} combined the use of the stochastic dominance technique with stochastic geometry results to study the stability of random networks, where terminals are located according to Poisson point processes. Conditions for queue stability were determined in \cite{stamatiou2010random} for a network with one and two classes of users.

The present work extends the results shown in \cite{stamatiou2010random}, expanding the formulation that describes the behavior of users in a random network with $N$ classes. We derive expressions in closed and simple forms for the necessary and sufficient conditions for the stability of the queues at the terminals of each class. More specifically, we have established the necessary and sufficient conditions relating user densities, transmission power levels and traffic intensities, that ensure the terminal queues of all classes will be stable. In addition, we show that, in the case of stable networks, the portion of the radio resource allocated to each class is well defined by a simple  expression relating its average delay, intensity of traffic, density of terminals, and the minimum acceptable signal-to-interference ratio (\emph{i.e.}, link quality). 

The network model adopted in the analysis presented here is based on a model used in \cite{stamatiou2010random, munari2016stability, nardelli2014throughput, zhou2016performance}, but with a key difference: while in these papers the separation distance between TX and RX terminals is assumed fixed, in our work we assume that this distance follows a Rayleigh distribution. The Rayleigh distribution assumption for the link distance was also used in other works, \emph{e.g.}, \cite{lin2014spectrum}. This assumption allowed us to obtain simple mathematical expressions relating traffic intensity, average delay, density of terminals and the required link quality of each class, when the network is stable. While the existence of an interplay among these parameters in a scenario where terminals share radio resources is intuitive, the formulation proposed here unveils this relationship, showing it in a simple way, allowing for insights into the trade-offs amongst key network parameters. 

Based on the formulation proposed here, we numerically evaluate the performance of a heterogeneous network with $N=2$ classes of terminals that share the same channel: cellular terminals, which access a base station or an access point, and D2D terminals, which communicate directly with each other. In particular, we consider the scenario where D2D terminals can access the channel used by cellular transmissions, but without causing excessive degradation to the performance of the cellular terminals. Using the formulation proposed here, we determine the maximum acceptable traffic intensity of D2D users that guarantees the average delay of cellular users does not exceed a given threshold.

The rest of the article is organized as follows: Section \ref{SystModel} describes the model used throughout the paper; in Section \ref{sec:one_user_class} we derive stability conditions and the mean delay for a simplified network, where all but one traffic class transmit dummy packets; Section \ref{sec:N-users} presents the main results of the paper, \emph{i.e.}, necessary and sufficient conditions for stability when we have $N$ interacting traffic classes, it also shows a simple expression for the stationary mean delay and the packet success probability; Section \ref{sec:application} applies the obtained results in two simplified scenarios: one scenario optimizes the transmission power of different traffic classes of D2D with different delay requirements sharing the same channel and the other analyses the performance of a D2D class sharing a channel with a cellular class (\textit{uplink}), where we set some delay requirements. The notations used in the paper are summarized in Table~\ref{tab:symbols}.

\section{System Model}
\label{SystModel}

\begin{table}[hbt]
   \centering
   \caption{Notations and Symbols Used in the Paper}
   \begin{tabular}{ll}
       \toprule
       \textbf{Symbol} & \textbf{Definitions/explanation} \\
       \midrule
        $\alpha\in(2,\infty)$	& path loss exponent \\
        $\delta\in(0,1)$		& $=2/\alpha$ \\
        $N$						& number of traffic classes \\
        $\cal{N}$				& the set $\{1,2,\dots,N\}$ \\
        $n\in\cal{N}$			& refers to the $n$-th traffic class \\
        $p_n$					& medium access probability \\
       	$a_n\in(0,1)$			& packet arrival rate per time slot \\
        $p_{s,n}$				& packet success probability \\
        $\theta_n$				& SIR threshold for successful communication \\
        $D_n\in(1,\infty)$		& average packet transmission delay \\
        $\overline{R}_n$		& mean transmission distance \\
        $P_n$					& transmission power \\
        $\Phi_n$				& Poisson point process for the sources \\
        $\lambda_n$				& density of $\Phi_n$ \\
        $\phi_n$ 				& $\triangleq 4\,\Gamma(1+2/\alpha)\,
        						  \Gamma(1-2/\alpha)\,
                        		  \overline{R}_n^2\,\theta_n^{2/\alpha}$ \\
        $||\cdot||$				& euclidean norm \\
        $\ind_A(x)$				& indicator function\\
       \bottomrule
   \end{tabular}
   \label{tab:symbols}
\end{table}

For each time slot $t \in \N$ and each traffic class $n \in \cal{N} = \{1,2,\dots,N\}$, we have a homogeneous Poisson point process (PPP) denoted by $\Phi_n(t) \subset \R^2$ of density $\lambda_n$, which represent the position of the sources. These PPP are independent from each other and from the past.
Each source of traffic class $n$ transmits with power $P_n$. The position of the sources are given by $\{ X_{i,n}(t) \}_i$, $i \in \N$, \textit{i.e.}, $\Phi_n(t) = \{ X_{i,n}(t) \}_i$. More precisely, for each time slot the position $X_{i,n}(t)$ of the source is reallocated following the high-mobility random walk model presented in \cite{baccelli2010stochastic}. The $i$-th source of traffic class $n$ communicates with a destination located at $Y_{i,n}(t)$. Thus, the distance between the $i$-th source of class $n$ and its destination is given by $R_{i,n}(t) = || X_{i,n}(t) - Y_{i,n}(t) ||$.
The random variables $\{Y_{i,n}(t)\}_t$ are defined such that $\{ R_{i,n}(t)\}_t$ are i.i.d. and distributed as Rayleigh, with mean transmission distance represented by $\overline{R_n}$. We have chosen this distribution, because it leads to simple results and it has a physical interpretation\footnote{Let $\Pi \subset \R^2$ be a PPP of density $\kappa$ and $R$ be the euclidean distance between the origin and the closest point of $\Pi$. Then, the p.d.f. of $R$ is $f_R(r) = 2\kappa\pi r \euler^{-\kappa\pi r^2}$, which is the Rayleigh density function. Furthermore,  $\E[R] = 1/\sqrt{4\kappa}$.}.
The occupation of the buffer at each source is represented by its queue length $\{ Q_{i,n}(t) \}$ of infinite capacity. The probability of a packet arrival at each queue is denoted by $a_n$ and the medium access probability by $p_n$. Within each slot, the first event to take place for each source with a non-empty queue is the medium access decision with probability $p_n$. If it is granted access and the SIR\footnote{We assume that thermal noise is negligible.} is greater than a threshold $\theta_n>0$, a packet is successfully transmitted and leaves the queue. Then, we have the arrival of the next packet with probability $a_n$. The last event to take place is the displacement of the sources and destinations. For more details about the order in which these events occur, see \cite{stamatiou2010random}. The main difference between this model and the one presented in \cite{stamatiou2010random} is that $R$ follows a Rayleigh distribution, instead of being constant.

The queue lengths of the source $i$, traffic class $n$ are Markov Chains represented by
\begin{equation}
	Q_{i,n}(t+1) = (Q_{i,n}(t) - D_{i,n}(t))_+ + A_{i,n}(t), \quad t \in \N,
\end{equation}
where $(\cdot)_+ = \max\{\cdot,0\}$, $A_{i,n}(t)$ are i.i.d. Bernoulli random variables of parameter $a_n$ and
\begin{equation*}
	D_{i,n}(t) = e_{i,n}(t)\,\ind_{\text{SIR}_{i,n}>\theta_{n}},
\end{equation*}
where $e_{i,n}(t)$ are i.i.d. Bernoulli random variables of parameter $p_n$, $\theta_n$ represents the SIR threshold for successful communication and the SIR of user $i$, traffic class $n$ is given by
\begin{equation}
	\text{SIR}_{i,n}(t) = \dfrac{P_n\,h_{i,n,i,n}(t)\,R_{i,n}(t)^{-\alpha}}
    {\sum_{(j,k)\neq(i,n)} P_k\,h_{j,k,i,n}(t)\,||X_{j,k}(t) - Y_{i,n}(t)||^{-\alpha} },
\end{equation}
where $h_{j,k,i,n}(t)$ are i.i.d. exponential distributed random variables of parameter 1 and represent the Rayleigh fading, $\alpha > 2$ is the path loss fading parameter.

\section{Single User Class Network Analysis}
\label{sec:one_user_class}

In this section, we analyze the behavior of one traffic class, given that all the other traffic classes transmit \emph{dummy} packets, \emph{i.e.}, their users always have packets to transmit. We are considering the buffer of only one traffic class. Without loss of generality let us study the first traffic class. From now on, for this section, whenever the subscript regarding the traffic class is omitted, we are referring to the first traffic class, \emph{i.e.}, $n=1$. This section is a stepping stone for the next and main section of the paper. It also compares the results of the modified model with the results of the original model \cite{stamatiou2010random}, where $R$ is constant.

\subsection{Stability Conditions and Stationary Analysis}
Sufficient and necessary conditions for stability of the buffers are shown in the following proposition.
\begin{proposition} \label{prop:stab_1user}
The queueing system $\{Q_{i}(t)\}$ is stable in the sense defined by \cite{szpankowski1994stability} if and only if
	\begin{equation} \label{eq:single_a}
		a < \dfrac{p}{1 + \phi\,(\lambda\,p + \zeta)},
	\end{equation}
where $\phi \triangleq 4\,\Gamma(1+2/\alpha)\,\Gamma(1-2/\alpha)\, \overline{R}^2\,\theta^{2/\alpha}$ and $\zeta \triangleq \sum_{n=2}^N (P_n/P_1)^{2/\alpha}\,\lambda_n\,p_n$. Then, the closure of arrival rates is given by
	\begin{equation}
		a \leq \dfrac{1}{1 + \phi\,(\lambda + \zeta)}.
	\end{equation}
\end{proposition}

\begin{proof}
Using the same arguments as in the proof of \cite[Proposition~1]{stamatiou2010random}, we have stability if and only if $\E[A_{i}(t)] < \E[D_{i}(t)]$ for the case where the first traffic class also transmits dummy packets (dominant network). Then, the effective PPP density of active sources from the $n$-th traffic class is $p_n\,\lambda_n$ and the result showed in \cite[Eq. (9)]{haenggi2009stochastic} gives that
\begin{equation}\label{eq:single_SIR}
\begin{split}
	\P(\text{SIR}_{i}(t) > \theta &\mid R_{i}(t) = r) \\
    &= \exp\left( - \pi \Gamma(1+2/\alpha) \Gamma(1-2/\alpha)\,\theta^{2/\alpha}\,r^2\,(\lambda\,p + \zeta) \right).
\end{split}
\end{equation}
The p.d.f. of $R_i(t)$ is given by $f_R(r) = 2\kappa\pi\,r\,\euler^{-\kappa\pi\,r^2}$, where $\kappa = 1/4\overline{R}^2$. Then, it is easy to calculate $\E[D_{i}(t)]$ by deconditioning Eq.~\eqref{eq:single_SIR}, which results in the right-hand side of \eqref{eq:single_a}. The left-hand side and the closure of \eqref{eq:single_a} is immediate.
\end{proof}

In the case where the system is stable, we can calculate the stationary probabilities, as showed in the following proposition.
\begin{proposition} \label{prop:ps}
When the system is stable, the stationary packet success probability is given by
	\begin{equation*}
		p_{s} = \dfrac{1 - \phi\,\lambda\,a}{1 + \phi\,\zeta}.
	\end{equation*}
\end{proposition}
\begin{proof}
	At steady state, the load at each queue is given by $\rho = a/(p\,p_{s})$ and the effective PPP density of active sources is given by $\lambda\,\rho\,p$. Following the same steps as in the proof of Proposition~\ref{prop:stab_1user}, we have that the stationary packet success probability is given by
    \begin{align*}
		p_{s} &= \P(\text{SIR}_{i}>\theta) \\
        	&= \int_0^\infty \P(\text{SIR}_{i}>\theta\mid R_i = r)\,f_R(r) \mathrm{d}r.
    \end{align*}
    Solving the above integral, we find that
    \begin{equation*}
    	p_{s} = \dfrac{1}{1+\phi\,(\lambda\,\rho\,p+\zeta)}
            = \dfrac{1}{1+\phi \left(\lambda\frac{a}{p_{s}}+\zeta\right)}.
    \end{equation*}
    Solving the above equation for $p_{s}$ ends the proof.
\end{proof}

\begin{proposition} \label{prop:delay_1-user}
	The stationary mean packet delay is given by
    \begin{equation*}
    	D = \dfrac{(1-a)(1+\phi\,\zeta)}{p-(1+\phi\,(\lambda\,p+\zeta))\,a},
    \end{equation*}
    which attains a minimum for a medium access probability $p=1$.
\end{proposition}
\begin{proof}
	From the proof of \cite[Proposition~3]{stamatiou2010random} we know that the stationary mean delay is given by $D=(1-a)/(p\,p_s-a)$. Then, the result follows directly from Proposition~\ref{prop:ps}.
\end{proof}

\begin{figure}[htb]
	\centering
	\input{./Plots/delay_1user.tex}	
	\caption{Delay $D$ as a function of the arrival rate of packets $a$ per time slot, at the optimum medium access probability ($p=1$) and $\zeta = 0$. Simulation results are shown in crosses. Dashed curves correspond to the model presented in \cite{stamatiou2010random}, where $R$ is constant.}
	\label{fig:delay_1user}
\end{figure}
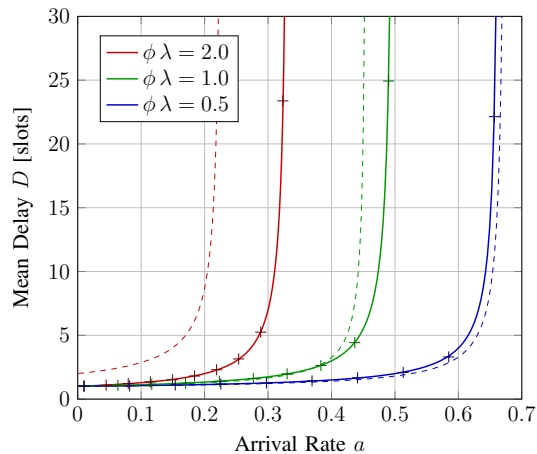

For comparison, $D$ is plotted as a function of $a$ in Fig.~\ref{fig:delay_1user} for the following parameters $\zeta = 0$ (only interference among users of the same class), $\phi\,\lambda = 0.5,1,2$, the medium access probability is chosen such that the delay is minimized ($p=1$) and we also plot the corresponding curves (dashed) for the model where $R$ is constant \cite{stamatiou2010random}. All the physical parameters $\lambda,\alpha,\theta,\overline{R}$ are set to be the same for each couple of curves. It is interesting to notice that, for low values of $\phi\,\lambda$ the variance from $R$ helps in the performance of the system, this is due to the fact that in an overcrowded system the possibility of $R$ having a probability to be small guarantees some successful transmissions, while in the other model, where $R$ is constant, this does not happen. We also performed simulations to ascertain our results (crossed points).
The average of links in the simulation were set to be 400.

\section{Multiple-Class Network}
\label{sec:N-users}

From now on, we consider a network with $N$ classes of users, all with buffer and we assume that the medium access probability for all traffic classes is equal to 1, to simplify the analysis. The motivation for this assumption, as can be seen in Section~\ref{sec:one_user_class}, is that it minimizes the delay and maximizes the stability region\footnote{Stability Region represents all the possible arrival rates, for which the system is stable.} for that case.

The following proposition presents the stationary success probability and delay, when transmitting a packet in a stable network. The results that guarantee stability are presented later in the paper.

\begin{proposition} \label{prop:psk}
	If the network is stable, then the stationary success probability and mean delay for each traffic class $n \in \cal{N}$ are given by
	\begin{align*}
    	p_{s,n} &= \left( 1 + \dfrac{\phi_n}{P_n^\delta}\, 
        \dfrac{\sum_j P_j^\delta\,\lambda_j\,a_j}
        {1 - \sum_j \phi_j\,\lambda_j\,a_j} \right)^{-1},\\
        D_n &= \dfrac{1 - a_n}{p_{s,n} - a_n},
    \end{align*}
    where the sums are taken over the set $\cal{N}$, $\phi_n \triangleq \Gamma(1+\delta)\,\Gamma(1-\delta)\,4\,\overline{R}_n^2\,\theta_n^{2/\alpha}$ and $\delta \triangleq 2/\alpha$.
\end{proposition}

\begin{proof}
	The delay follows directly from the proof of \cite[Proposition~3]{stamatiou2010random}. At steady state, we have that the effective PPP density of active sources for each traffic class is $\lambda_n\,a_n/p_{s,n}$. Then, using the result from \cite[Eq.~(9)]{haenggi2009stochastic}, one can show that
    \begin{align} \label{eq:aux_psk}
    	p_{s,n} &= \P(\text{SIR}_{i,n}>\theta_n) \nonumber\\
        	&= \int_0^\infty \P(\text{SIR}_{i,n}>\theta_n
            \mid R_{i,n}(t) = r)\,f_R(r)\,\mathrm{d}r \nonumber\\
            &= \left( 1 + \dfrac{\phi_n}{P_n^\delta} 
        \sum_j P_j^\delta \dfrac{\lambda_j\,a_j}{p_{s,j}} \right)^{-1},
    \end{align}
    which can be rearranged as
    \begin{equation} \label{eq:psk_equivalence}
    	\dfrac{P_n^\delta}{\phi_n} \left( \dfrac{1-p_{s,n}}{p_{s,n}}\right) = 
        \sum_j P_j^\delta\dfrac{\lambda_j\,a_j}{p_{s,j} } .
    \end{equation}
    Note that the right-hand side does not depend on $n$. Then, for all $j$, we can write
    \begin{equation} \label{eq:Pi_Pk}
    	\dfrac{P_j^\delta}{\phi_j} \left( \dfrac{1-p_{s,j}}{p_{s,j}}\right) = 
        \dfrac{P_n^\delta}{\phi_n} \left( \dfrac{1-p_{s,n}}{p_{s,n}}\right).
    \end{equation}
    For each $j$, we can solve the above equation for $p_{s,j}$ and plug into the sum of Eq.~\eqref{eq:aux_psk}. Then, we can solve it for $p_{s,n}$, which ends the proof.
\end{proof}

\begin{lemma} \label{lem:identity_1}
	If the network is stable, then the following identity holds (at steady state),
    \begin{equation*}
    	\sum_{n\in\cal{N}} \phi_n\,\lambda_n\,\dfrac{D_n}{D_n-1}\,\dfrac{a_n}{1-a_n} = 1.
    \end{equation*}
    Furthermore,
    \begin{equation*}
    	\dfrac{\phi_j}{P_j^\delta}
        \left( \dfrac{D_j}{D_j-1}\,\dfrac{1}{1-a_j} - 1\right) =
        \dfrac{\phi_k}{P_k^\delta} \left( \dfrac{D_k}{D_k-1}\,
        \dfrac{1}{1-a_k} - 1 \right)
        \quad \forall\,j,k\in\cal{N}.
    \end{equation*}
\end{lemma}

\begin{proof}
	We start with the terms of the sum,
	\begin{align*}
		\phi_n\,\lambda_n\,\dfrac{D_n}{D_n-1}\,\dfrac{a_n}{1-a_n}
        &\stackrel{\text{(i)}}{=} \phi_n\lambda_n\dfrac{a_n}{1-p_{s,n}}\\ 
		&= P_n^\delta \dfrac{\lambda_n\,a_n}{p_{s,n}} \left( 
        \dfrac{\phi_n}{P_n^\delta} \dfrac{p_{s,n}}{1-p_{s,n}} \right)\\
        &\stackrel{\text{(ii)}}{=} \dfrac{ P_n^\delta\dfrac{\lambda_n\,a_n}{p_{s,n}} }
        { \sum_j P_j^\delta\dfrac{\lambda_j\,a_j}{p_{s,j}} },
	\end{align*}
    where (i) comes from Proposition~\ref{prop:psk} and (ii) comes from Eq.~\eqref{eq:psk_equivalence}. Summing over $\cal{N}$ ends the proof of the first identity. For the second relation, we use Proposition~\ref{prop:psk} once again to find that
    \begin{equation*}
    	\dfrac{\phi_n}{P_n^\delta} \left( \dfrac{D_n}{D_n-1}\,\dfrac{1}{1-a_n} - 1\right) = \dfrac{\phi_n}{P_n^\delta} \dfrac{p_{s,n}}{1-p_{s,n}}.
    \end{equation*}
    Comparing this expression with Eq.~\eqref{eq:Pi_Pk} ends the proof.
\end{proof}

	Lemma~\ref{lem:identity_1} is an elegant form to see that a channel is a limited resource regarding traffic density and delay. Let us rewrite the identity in terms of physical parameters,
    \begin{equation} \label{eq:physical}
    	\sum_{n=1}^N 4\,\lambda_n\,\overline{R}_n^2\,\theta_n^{2/\alpha}\,\dfrac{D_n}{D_n-1}\,\dfrac{a_n}{1-a_n} = \dfrac{\sin(2\pi/\alpha)}{2\pi/\alpha},
    \end{equation}
    where we used Euler's reflection formula. Note that $\frac{a_n}{1-a_n}$ and $\frac{\sin(2\pi/\alpha)}{2\pi/\alpha}$ are monotonic increasing functions and $\frac{D_n}{D_n-1}$ is a monotonic decreasing function. The right hand-side of Eq.~\eqref{eq:physical} can be seen as a resource available to the users of the channel. The larger the path loss exponent $\alpha$, the larger (smaller) the terms $\lambda_n$, $\overline{R}_n$, $\theta_n$, $a_n$ ($D_n$) can be. A possible modification is to make a direct exchange between decreasing the delay $D_n$ and decreasing the arrival rate of packets $a_n$ (by controlling the ratio of transmit power levels as it is showed in Section~\ref{ssec:cel_D2D}), such that the term $\frac{D_n}{D_n-1}\,\frac{a_n}{1-a_n}$ remains constant; or else increase the arrival rate of packets and decrease the number of users, such that the term $\lambda_n\,\frac{a_n}{1-a_n}$ remains constant; we can also exchange quantities among the terms of traffic class $k$ and $\ell$, such that the sum $\lambda_k\,\overline{R}_k^2\,\theta_k^{2/\alpha}\,\frac{D_k}{D_k-1}\,\frac{a_k}{1-a_k}+\lambda_\ell\,\overline{R}_n^2\,\theta_n^{2/\alpha}\,\frac{D_\ell}{D_\ell-1}\,\frac{a_\ell}{1-a_\ell}$ remains constant; and so on.

\begin{lemma} \label{lem:sufficient}
	A necessary and sufficient condition for the network stability is that $\bm{a} \in \bigcup_{\nu\in\cal{P}} \cal{S}_\nu$, where $\cal{P}$ is the space of all bijective functions from $\cal{N}$ to $\cal{N}
    $ and
    \begin{equation*}
    	\cal{S}_{\nu} = \left\lbrace \bm{a} \in [0,1]^N \bigm\vert
        \dfrac{\phi_{\nu(n)}}{P_{\nu(n)}^\delta}\,\dfrac{a_{\nu(n)}}{1-a_{\nu(n)}} < 
        \dfrac{ 1 - \sum_{k=1}^{n-1} \phi_{\nu(k)}\,\lambda_{\nu(k)}\,a_{\nu(k)} }
        {\sum_{k=1}^{n-1} P_{\nu(k)}^{\delta}\,\lambda_{\nu(k)}\,a_{\nu(k)} +
        \sum_{k=n}^N P_{\nu(k)}^\delta\,\lambda_{\nu(k)}}        
        ~ \forall n \in \cal{N} \right\rbrace,
    \end{equation*}
    with the convention $\sum_{k=1}^0 \cdot = 0$.
\end{lemma}

\begin{proof}
	See Appendix~\ref{app:lemm_stab}.
\end{proof}

The following theorem presents a simple form of stating Lemma~\ref{lem:sufficient} and relates (in the proof) the stability condition with the stationary mean delay in Lemma~\ref{lem:identity_1}.

\begin{theorem} \label{th:stability}
	The system network is stable if and only if $\bm{a}\in\cal{R}$, where
    \begin{align*}
    	\cal{R}
        &\triangleq \left\lbrace \bm{a}\in[0,1]^N \bigm\vert 
        \dfrac{\phi_n}{P_n^\delta}\,\dfrac{a_n}{1-a_n} < 
        \dfrac{1 - \sum_{k} \phi_k\,\lambda_k\,a_k}
        {\sum_{k} P_k^{\delta}\,\lambda_k\,a_k}
        \quad \forall n \in \cal{N}  \right\rbrace \\
        &= \left\lbrace \bm{a}\in[0,1]^N \bigm\vert 
        \dfrac{\phi_n}{P_n^\delta}\,\dfrac{a_n}{1-a_n} < 
        \dfrac{1 - \sum_{k \neq n} \phi_k\,\lambda_k\,a_k}
        {P_n^\delta\,\lambda_n + \sum_{k \neq n} P_k^{\delta}\,\lambda_k\,a_k}
        \quad \forall n \in \cal{N}  \right\rbrace.
    \end{align*}
\end{theorem}

\begin{proof}
	See Appendix~\ref{app:th_stab}. The proof simply shows that $\cal{R} = \bigcup_{\nu\in\cal{P}} \cal{S}_\nu$.
\end{proof}

From the proof of Theorem~\ref{th:stability}, it is clear that for an arbitrary choice of the stationary mean delays $\bm{D}\in(1,\infty)^N$, it is possible to determine a vector of arrival rates $\bm{a}\in\cal{R}$, such that the specified mean delays are achieved.

The following corollary establishes a simple stability result, which will be useful in the Section~\ref{sec:application}, where we deal with an optimization problem regarding the transmit powers.

\begin{corollary} \label{cor:stab}
	There exist $P_1,P_2,\dots,P_N \in \R_+$ such that the network is stable if and only if
    \begin{equation*}
    	\sum_{n\in\cal{N}} \phi_n\,\lambda_n\,\dfrac{a_n}{1-a_n} < 1.
    \end{equation*}
\end{corollary}
\begin{proof}
	If for some $P_1,P_2,\dots,P_N$ the system is stable, then from Theorem~\ref{th:stability}, $\bm{a}\in\cal{R}$ and from Lemma~\ref{lem:identity_1} we have that
    \begin{align*}
    	1 
        &= \sum_{n\in\cal{N}} \phi_n\,\lambda_n\,
        \dfrac{D_n}{D_n-1}\,\dfrac{a_n}{1-a_n} \\
        &> \sum_{n\in\cal{N}} \phi_n\,
        \lambda_n\,\dfrac{a_n}{1-a_n}.
    \end{align*}
	On the other hand, if we know that $\sum_{n\in\cal{N}} \phi_n\,\lambda_n\frac{a_n}{1-a_n}<1$, then we can choose $D_1,D_2,\dots,D_N\in(1,\infty)$ such that $\sum_{n\in\cal{N}} \phi_n\,\lambda_n\frac{D_n}{D_n-1}\frac{a_n}{1-a_n}=1$. It is easy to see that we can find $P_1,P_2,\dots,P_N$ such that $\bm{a}\in\cal{R}$ in \eqref{eq:set_identity_1} and, again from Theorem~\ref{th:stability}, the system is stable.
\end{proof}

\section{Interpretation and Application of the Results}
\label{sec:application}

In this section, we present some numerical results using the proposed  formulation, applied to scenarios of different classes of D2D terminals and cellular terminals sharing a radio channel.

\subsection{Optimization problem}

Let us consider the scenario with $N$ classes of D2D terminals sharing a channel. Each class may represent a particular user application, with each application having a different delay requirement in the network. For instance, applications such as Tactile Internet \cite{fettweis2014tactile} or V2V have a more restrictive delay requirement than video streams. Let us suppose we are interested in adjusting the transmit power of each traffic class, such that the weighted average delay among all classes is minimized. This problem may be addressed as follows. For fixed arrival rates $\bm{a}$ that satisfies Corollary~\ref{cor:stab}, let us minimize the delays $\bm{D}$ by controlling the ratio between the transmission powers $\bm{P}$. Since each traffic class may require different response times, let us weight the optimization problem with the vector $(c_1,c_2,\dots,c_N) \in \R_+^N$. The larger the coefficient of a class, the smaller the mean delay to deliver packets for that class. Then, we have
\begin{equation} \label{eq:opt_prob}
	\min_{\bm{P}\in\R_+^N}~\sum_{n\in\cal{N}} c_n\,D_n,
\end{equation}
where $D_n$ is given by Proposition~\ref{prop:psk}. Note that as thermal noise is not considered in our model, we have a degree of freedom for the solution $\bm{P}^*$. 

\begin{proposition} \label{prop:opt}
	The minimum of the optimization problem \eqref{eq:opt_prob} is attained by $\beta\bm{P}^*$, where $\beta$ is any positive real constant and for $n\in\cal{N}$,
    \begin{equation*}
        {P_n^*}^\delta
        	= \dfrac{ \phi_n\,\dfrac{a_n}{1-a_n} }
        	{ 1-\displaystyle\sum_{k\in\cal{N}}\phi_k\,\lambda_k\,
            \dfrac{a_k}{1-a_k} } + \dfrac{ \sqrt{\dfrac{c_n\,\phi_n}
            {\lambda_n\,a_n\,(1-a_n)}} }{\displaystyle\sum_{k\in\cal{N}}
            \sqrt{c_k\,\phi_k\,\lambda_k\,\dfrac{a_k}{1-a_k}}}.
    \end{equation*}
\end{proposition}

\begin{proof}
	This can be proved by using Karush–Kuhn–Tucker conditions \cite[Section~3.3.1]{bertsekas1999nonlinear}.
\end{proof}

It is interesting to note that if $c_n = \phi_n\,\lambda_n\,\frac{a_n}{1-a_n}$, then the optimum is attained when $D_1 = D_2 = \cdots = D_N = \left( 1 - \sum_k \phi_k\,\lambda_k\,\frac{a_k}{1-a_k} \right)^{-1}$.

\begin{figure}[t]
\centering
\input{./Plots/Opt_Delay.tex}
\caption{Delays of two-class D2D network, for $a_1=a_2=0.7$, $\phi_1\,\lambda_1 = \phi_2\,\lambda_2 = 0.15$, and $c_1 = 1$.}
\label{TwoClassD2D_delay}
\end{figure}
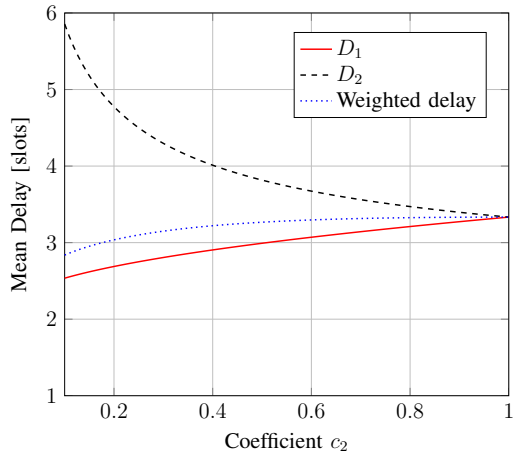


As an example, let us consider a two-class D2D network, where Class 1 has a more restrictive delay requirement than Class 2, such that we choose $c_1 \ge c_2$. Figure \ref{TwoClassD2D_delay} show the delays $D_n$, $n = 1,2$, for $\frac{1}{10} \leq c_2 \leq 1$, $c_1 = 1$ and $a_1=a_2=0.7$. As expected, due to the stricter delay requirement of Class 1 and the symmetry between both classes, the optimization resulted in $D_1 < D_2$ and $P_1^* > P_2^*$.



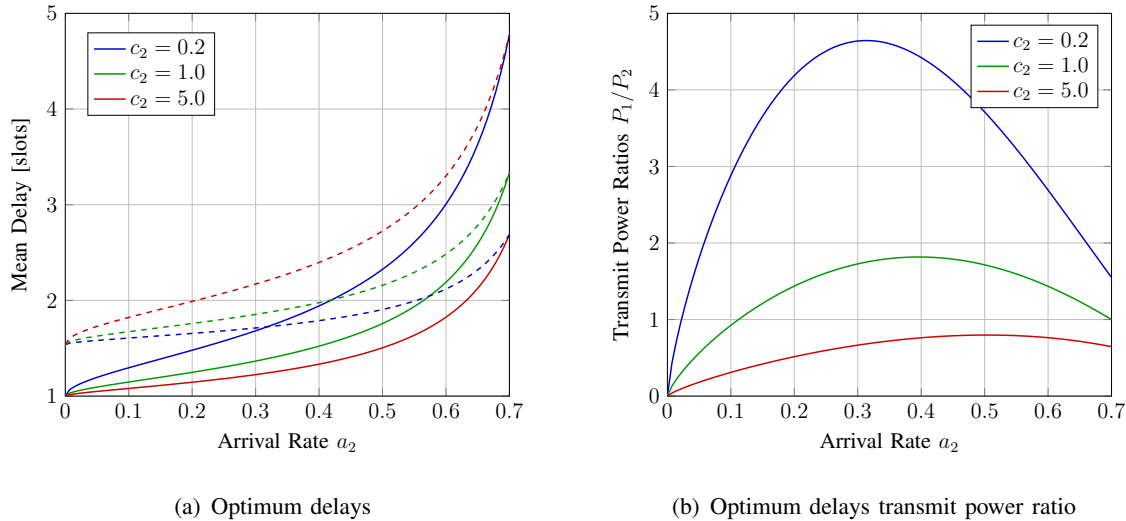
\begin{figure}[htb]
\centerline{
\subfigure[Optimum delays]{
\input{./Plots/Opt_Delay_a.tex}
\label{TwoClassD2D_delay_arrival}}
\hfil
\subfigure[Optimum delays transmit power ratio]{
\input{./Plots/Opt_RatioPower.tex}
\label{TwoClassD2D_pwr}}
}
\caption{These figures represent the optimization of a 2-class D2D network with the following parameters: $a_1=0.7$, $\phi_1\,\lambda_1 = \phi_2\,\lambda_2 = 0.15$ and $c_1 = 1$.
In the left figure, the dashed curve is $D_1$ and the full curve is $D_2$.}
\label{fig:opt}
\end{figure}

Another interesting example is to consider two classes with the same $\lambda\,\phi$ parameter, but with different arrival rates. Figures \ref{TwoClassD2D_delay_arrival} and \ref{TwoClassD2D_pwr} show the delays and the transmit power ratio as a function of the arrival rate of the second class, respectively. As expected, in Fig.~\ref{TwoClassD2D_delay_arrival}, when $a_2$ increases, both delays increase and when $c_2\ge c_1$, $D_2$ (full curve) tends to remain below $D_1$ (dashed curve). It is worth noting that the curve of the transmit powers in Fig.~\ref{TwoClassD2D_pwr} has a maximum. A possible explanation of this interesting behavior is that when $a_2$ is small, the second class rarely causes interference in the first class, then $P_1/P_2 \approx 0$ is the best choice to minimize the delays. As $a_2$ increases, it is necessary to increase the relative transmit power of Class 1, since the interference of Class 2 in Class 1 increases. However, when $a_2$ is large enough, the packet success probability of Class 2 becomes a concern, therefore it is necessary to decrease the interference from Class 1, thus decreasing the ratio $P_1/P_2$.

\subsection{Cellular and D2D}
\label{ssec:cel_D2D}

Let the first and second traffic classes represent the D2D and the cellular, respectively. For the cellular, we consider the \textit{uplink} transmission, which is closer to the proposed model, since the base stations do not move, only the users move and the model uses a high-mobility PPP. Furthermore, we must disregard temporal correlation to adequate to the model assumptions.

The maximum arrival rate for the D2D user, when we are able to control the transmission power is given by Proposition~\ref{prop:max_a1} and a numerical example is showed in Fig.~\ref{fig:max_a1}, where the quantity $\Psi_n \triangleq \phi_n\, \lambda_n\, \frac{D_n}{D_n-1}\, \frac{a_n}{1-a_n} \geq 0$ measures the use of the channel by the $n$-th traffic class in the sense presented by Lemma~\ref{lem:identity_1}, where we have that $\sum_{n\in\cal{N}}\Psi_n = 1$. Then, it is natural to think that the $n$-th traffic class uses a percentage $\Psi_n$ of the channel.

\begin{proposition} \label{prop:max_a1}
Given the arrival rate $a_2$ and the constraints $D_1 \in (1,D_1^*]$ and $D_2 \in (1,D_2^*]$, the possible arrival rates for the first traffic class, over all $P_1,P_2 \in \R_+$, such that the system is stable, is given by
\begin{equation*}
	a_1 \leq \left( 1 + \dfrac{ \phi_1\,\lambda_1\,\frac{D_1^*}{D_1^*-1} }
    { 1 - \Psi_2^* } \right)^{-1},
\end{equation*}
when $\Psi_2^* < 1$, where $\Psi_2^* \triangleq \phi_2\,\lambda_2\,\frac{D_2^*}{D_2^*-1}\,\frac{a_2}{1-a_2}$. We can achieve equality with
\begin{equation*}
	\frac{\phi_1}{\phi_2}\frac{P_2^\delta}{P_1^\delta}
    = \frac{\frac{\Psi_2^*}{\phi_2\,\lambda_2} + \frac{1}{D_2^*-1}}
        {\frac{1-\Psi_2^*}{\phi_1\,\lambda_1}+\frac{1}{D_1^*-1}}.
\end{equation*}
\end{proposition}
\begin{proof}
	Follows from Lemma~\ref{lem:identity_1} and Theorem~\ref{th:stability}.
\end{proof}



\begin{figure}[htb]
\centerline{
\subfigure[Maximum arrival rate]{
\input{./Plots/max_a1.tex}
\label{fig:max_a1}}
\hfil
\subfigure[Transmit power ratio]{
\input{./Plots/P1P2.tex}
\label{fig:P1P2}}
}
\caption{Left figure shows the maximum arrival rate achievable for the first traffic class (D2D), such that the constraints of Proposition~\ref{prop:max_a1} are satisfied; $\Psi_2$ represents the use of the channel by the second traffic class (cellular); we used $\phi_1\,\lambda_1 = 1$. Right figure shows the transmit power ratio to achieve the maximum $a_1$; we used $\phi_1\,\lambda_1 = \phi_2\,\lambda_2 = 1$ and $D_2^* = 3$ [slots].}
\label{fig:2_class}
\end{figure}
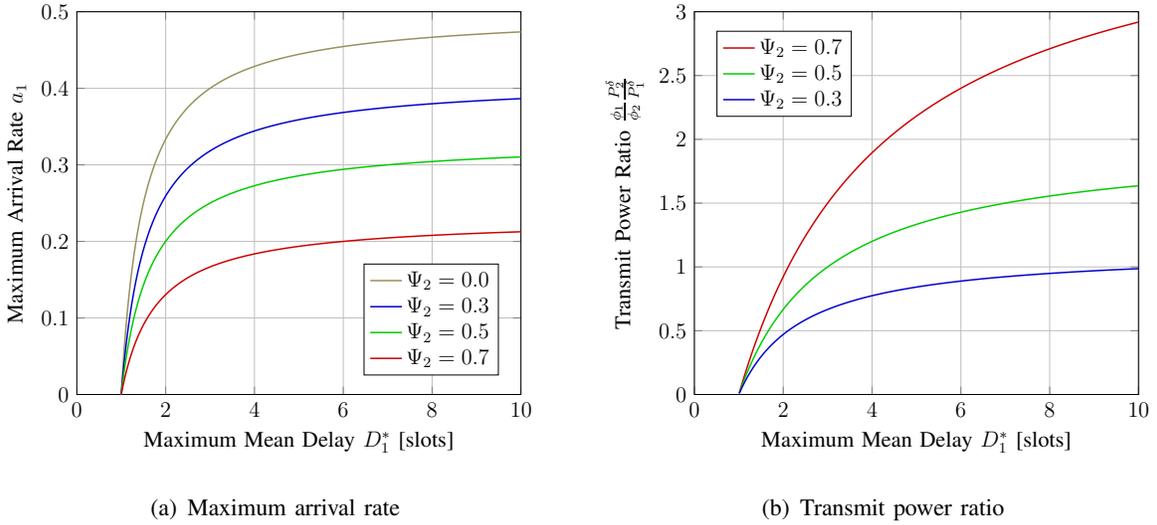

As expected, Fig.~\ref{fig:max_a1} shows that as the use of the channel by the cellular increases or as the maximum delay constraint of the D2D decreases, the maximum arrival rate for the D2D decreases. Furthermore, as the D2D delay constraint increases, the smaller the impact of this change over the maximum arrival rate permitted. This analysis agrees with the simple equation deduced in Lemma~\ref{lem:identity_1} that at steady state $\Psi_1+\Psi_2 = 1$, \textit{i.e.}, we may divide a percentage $\Psi_1$ of the use of the channel for the cellular and the other percentage $\Psi_2$ for the D2D. Then, for the $\Psi$ available, we may choose the parameters of performance $a$ and $D$, such that we preserve the identity $\phi\, \lambda\, \frac{D}{D-1}\, \frac{a}{1-a} = \Psi$. The quantity $\phi$ is a constant related to the mean link distance of transmission and the SIR threshold for successful communication. Therefore, the relation by $a$ and $D$ is determined by the term $\phi\,\lambda$, which is proportional to the quantities $\lambda\,\overline{R}^2\,\theta^{2/\alpha}$.

In order to attain the maximum arrival rate for the D2D, it is necessary to have the transmission power ratios presented in Proposition~\ref{prop:max_a1}. In Fig.~\ref{fig:P1P2}, it is shown this ratio as a function of the maximum delay for the D2D for some values of $\Psi_2$, which is the percentage of the channel used by the cellular. As expected, as we increase the maximum delay for the D2D or as we decrease the use of the channel by the D2D, the smaller the relative power transmission required. It is remarkable that, again, as the maximum delay constraint increases, the smaller the impact over the power transmission ratio. Differently from the quantities $a$ and $D$, the power transmission ratio is not simply determined by the product $\lambda\,\overline{R}^2\,\theta^{2/\alpha}$, we need to know the values of $\overline{R}^2\,\theta^{2/\alpha}$ and $\lambda$ separately.

\section{Conclusions}
In this paper, it is proposed a modified model to study the stability and delay of slotted Aloha in Poisson networks. The main modification of the model presented in the paper with respect to other models presented in the literature is to consider an i.i.d. Rayleigh distribution for the distance of the link between source and destination. This provided tractability to the model: we derived necessary and sufficient conditions for stability in a network with $N$ user-classes; we also provided simple closed-form expressions for the packet success probability and mean delay. As shown by the results in the paper, the advantage of using this model as a base to model other network effects is its analytical tractability. For example, we were able to derive simple conditions to verify the stability of a network with undetermined transmit powers (see Corollary~\ref{cor:stab}). We also solved (analytically) an optimization problem regarding the minimization of the delays in a network (see Proposition~\ref{prop:opt}); this result was applied to a numerical example involving a D2D network and it showed interesting insights about the optimum transmit power of the user-classes.

\appendices
\section{Proof of Lemma~\ref{lem:sufficient}}
\label{app:lemm_stab}
\begin{proof}
	Using the concept of stochastic dominance \cite[Section~2.1.2]{kompella2014stable}, it is possible to derive necessary and sufficient conditions for stability. In the dominant network, all the traffic classes in the set $\cal{D}\subset\cal{N}$ transmit dummy packets. If the dominant network is stable, then the original network is stable. On the other hand, if the queues of the traffic classes in $\cal{D}$ are not empty in the original network, then this system behaves exactly as the dominant network (both systems are \emph{indistinguishable} \cite[Section~3.2]{szpankowski1994stability}). Therefore, if the dominant network is unstable, then the original network will be unstable as well. In order to have necessary and sufficient conditions, we must perform this verification for each $\cal{D} \subset \cal{N}$.
    
    Let us start with $\cal{D}=\cal{N}$, \emph{i.e.}, all users transmit dummy packages. For each step of the verification, we remove the stable traffic class from the set $\cal{D}$. This procedure repeats until the set $\cal{D}$ becomes empty. In order to attain stability of the dominant network we must have an incoming packet probability smaller than the success probability \cite{loynes1962stability}. A sufficient condition for the first traffic class stability is, for any queue $i$ of this class (by symmetry),
    \begin{equation*}
    	a_1 < \P(\widetilde{\text{SIR}}_{i,1} > \theta_1) = \left( 1 + \dfrac{\phi_1}{P_1^\delta} \sum_{k=1}^N P_k^\delta\,\lambda_k \right)^{-1},
    \end{equation*}
    where $\widetilde{\text{SIR}}$ represents the signal-interference ratio in the dominant network. This guarantees stability for the first traffic class. Let us remove it from the set $\cal{D}$. Then, we calculate the stationary success probability of the first traffic class $\widetilde{p}^{(1)}_{s,1}$ for this dominant network. At steady state, we have
    \begin{equation*}
    	\widetilde{p}^{(1)}_{s,1} = \left( 1 + \dfrac{\phi_1}{P_1^\delta} \left( P_1^\delta\,\lambda_k\dfrac{a_1}{\widetilde{p}^{(1)}_{s,1}} + \sum_{k=2}^N P_k^\delta\,\lambda_k \right) \right)^{-1},
    \end{equation*}
    which can be solved for $\widetilde{p}^{(1)}_{s,1}$,
    \begin{equation*}
    	\widetilde{p}^{(1)}_{s,1} = \dfrac{1 - \phi_1\,\lambda_1\,a_1}{1 + \dfrac{\phi_1}{P_1^\delta} \sum_{k=2}^N P_k^\delta\,\lambda_k}.
    \end{equation*}
    The next step is to verify the conditions of stability for the second traffic class, when the first traffic class is at steady state. After that, we remove the second traffic class from the set $\cal{D}$ and calculate the stationary success probability of the two stable traffic classes in the dominant network. We repeat these steps until we remove all traffic classes, \textit{i.e}, $\cal{D} = \{\}$. We show this by induction; we suppose stability of the traffic classes $1,2,\dots,j-1$. Let $\cal{D} = \{j,j+1,\dots\,N\}$; the $j$-th traffic class is stable, given that all the traffic classes in $\cal{N}\setminus\cal{D}$ are stable, when
    \begin{equation} \label{eq:aux_aj}
    	a_j < \P(\widetilde{\text{SIR}}_{i,j} > \theta_j) = \left( 1+\dfrac{\phi_j}{P_j^\delta}
        \left( \sum_{k=1}^{j-1} P_k^\delta\,\lambda_k\,\dfrac{a_k}
        {\widetilde{p}^{(j)}_{s,k}} + \sum_{k=j}^N P_k^\delta\,\lambda_k \right) \right)^{-1},
    \end{equation}
    where $\widetilde{p}^{(j)}_{s,k}$ is the $k$-th traffic class success probability ($1 \leq k < j$) at steady state in the dominant network at the $j$-th step. To calculate this probability, we must solve the following system of equations. For $k \in \{1,2,\dots,j-1\}$
    \begin{equation*}
    	\widetilde{p}^{(j)}_{s,k} = \left( 1 + \dfrac{\phi_k}{P_k^\delta}
        \left( \sum_{\ell = 1}^{j-1} P_\ell^\delta\,\lambda_\ell\,
        \dfrac{a_\ell}{\widetilde{p}^{(j)}_{s,\ell}} + 
        \sum_{\ell = j}^{N} P_\ell^\delta\,\lambda_\ell \right) \right)^{-1}.
    \end{equation*}
    Using an analogous approach as the one presented in the proof of Proposition~\ref{prop:psk}, we have that
    \begin{equation*}
    	\widetilde{p}^{(j)}_{s,k} =
        \left( 1 + \dfrac{\phi_k}{P_k^\delta} 
        \dfrac{ \sum_{\ell=1}^{j-1} P_\ell^\delta\,\lambda_\ell\,a_{\ell} +
        \sum_{\ell=j}^{N} P_\ell^\delta\,\lambda_\ell }
        { 1 - \sum_{\ell=1}^{j-1} \phi_{\ell}\,\lambda_\ell\,a_{\ell}} 
        \right)^{-1}, \quad k \in \{1,2,\dots,j-1\}.
    \end{equation*}
    Comparing the last two equations, it is easy to see that
    \begin{equation*}
    	\sum_{\ell = 1}^{j-1} P_\ell^\delta\,\lambda_\ell\,
        \dfrac{a_\ell}{\widetilde{p}^{(j)}_{s,\ell}} + 
        \sum_{\ell = j}^{N} P_\ell^\delta\,\lambda_\ell = 
        \dfrac{ \sum_{\ell=1}^{j-1} P_\ell^\delta\,\lambda_\ell\,a_{\ell} +
        \sum_{\ell=j}^{N} P_\ell^\delta\,\lambda_\ell }
        { 1 - \sum_{\ell=1}^{j-1} \phi_{\ell}\,\lambda_\ell\,a_{\ell} } .
    \end{equation*}    
    Finally, we can use this result to rewrite Eq.~\eqref{eq:aux_aj} as
    \begin{equation*}
    	\dfrac{\phi_j}{P_j^\delta}\,\dfrac{a_j}{1-a_j} <
        \dfrac{1 - \sum_{k=1}^{j-1} \phi_k\,\lambda_k\,a_k}
        {\sum_{k=1}^{j-1} P_k^\delta\,\lambda_k\,a_k +
        \sum_{k=j}^N P_k^\delta\,\lambda_k}, \quad j \in \cal{N}.
    \end{equation*}
    This concludes the proof, since the extension for the other partitions of $\cal{N}$ is analogous.
\end{proof}

\section{Proof of Theorem~\ref{th:stability}}
\label{app:th_stab}
\begin{proof}
    First, let us show that the set related to Lemma~\ref{lem:identity_1} (by taking all the possible delays) corresponds to the region $\cal{R}$ defined in the statement of the theorem, \textit{i.e.}, let us show that $\cal{R}=\cal{R}'$, where
\begin{equation} \label{eq:set_identity_1}
    \begin{split}
    	\cal{R}' \triangleq
        &\bigcup_{\bm{D}\in(1,\infty)^N} 
        	\left\lbrace \bm{a}\in[0,1]^N
        	\bigm\vert \sum_{n\in\cal{N}} \phi_n\,\lambda_n\,
        	\dfrac{D_n}{D_n-1}\,\dfrac{a_n}{1-a_n} = 1, \right.\\
        &\left. \quad\dfrac{\phi_j}{P_j^\delta} \left( 
			\dfrac{D_j}{D_j-1}\,\dfrac{1}{1-a_j} - 1\right) =        
        	\dfrac{\phi_k}{P_k^\delta} \left(
        	\dfrac{D_k}{D_k-1}\,\dfrac{1}{1-a_k} - 1\right)
        	\quad \forall\,j,k\in\cal{N} \right\rbrace.
	\end{split}
\end{equation}
     This can be seen by manipulating the equations that define the set in \eqref{eq:set_identity_1}. Let us start with
    \begin{equation*}
    	\sum_{k} \phi_k\,\lambda_k\,\dfrac{D_k}{D_k-1}\,
        \dfrac{a_k}{1-a_k} = 1,
    \end{equation*}
    which can be rewritten as
    \begin{equation*}
        \sum_{k} P_k^\delta\,\lambda_k\,a_k\dfrac{\phi_k}{P_k^\delta}
        \left(\dfrac{D_k}{D_k-1}\dfrac{1}{1-a_k} - 1\right)\,
        = 1 - \sum_k \phi_k\,\lambda_k\,a_k
    \end{equation*}
	Then, using the other equations in \eqref{eq:set_identity_1}, we have that
    \begin{equation*}
    	\dfrac{\phi_n}{P_n^\delta}\left(\dfrac{D_n}{D_n-1}
        \dfrac{1}{1-a_n} - 1\right)\sum_{k} P_k^\delta\,\lambda_k\,a_k\,
        	= 1 - \sum_k \phi_k\,\lambda_k\,a_k.
    \end{equation*}
    Since $\frac{D_n}{D_n-1} \in (1,\infty)$, then for each $n\in\cal{N}$, we must have
    \begin{equation} \label{eq:aux_R}
        \dfrac{\phi_n}{P_n^\delta}\,\dfrac{a_n}{1-a_n}
        < \dfrac{\phi_n}{P_n^\delta}\left(\dfrac{D_n}{D_n-1}
        \dfrac{1}{1-a_n} - 1\right)
        = \dfrac{1 - \sum_{k} \phi_k\,\lambda_k\,a_k}
        {\sum_{k} P_k^{\delta}\,\lambda_k\,a_k}.
    \end{equation}
    Therefore, $\cal{R}'\subset\cal{R}$. On the other hand, when $D_n$ varies continually from 1 to $\infty$, $a_n$ varies continually from 0 to the maximum value respecting Eq.~\eqref{eq:aux_R}, which means that $\cal{R}\subset\cal{R}'$. Then, $\cal{R}' = \cal{R}$.
    
    Now, let us prove that $\cal{R}' \subset \bigcup_{\nu\in\cal{P}} \cal{S}_\nu$. Note that the stability region of Lemma~\ref{lem:sufficient} demands that at least one $a_n$ ($n\in\cal{N}$) satisfies
\begin{equation} \label{eq:aux_stab_first}
	\dfrac{\phi_n}{P_n^\delta}\,\dfrac{a_n}{1-a_n} <
    \dfrac{1}{\sum_{k=1}^N P_n^\delta\,\lambda_n}.
\end{equation}
	Let us show that $\cal{R}'$ requires the same restriction by contradiction. Suppose that there exist $\bm{a}\in\cal{R}'$ such that
\begin{equation*}
    \dfrac{\phi_n}{P_n^\delta}\,\dfrac{a_n}{1-a_n} >
    \dfrac{1}{\sum_{k=1}^N P_k^\delta\,\lambda_k}
    \quad \forall n \in \cal{N}.
\end{equation*}
If $\bm{a}\in\cal{R}'$, then using \eqref{eq:set_identity_1} and the above inequality, we have that
\begin{align*}
	1	&= \sum_{n=1}^N P_n^\delta\,\lambda_n\,\dfrac{D_n}{D_n-1}\,
    \dfrac{\phi_n}{P_n^\delta}\,\dfrac{a_n}{1-a_n} \\
    	&> \dfrac{\sum_{n=1}^N P_n^\delta\,\lambda_n\,\dfrac{D_n}{D_n-1}}
    {\sum_{k=1}^N P_k^\delta\,\lambda_k} \\
    	&> 1.
\end{align*}
The last inequality comes from the fact that $D_n/(D_n-1)>1$, if $D_n\in(1,\infty)$. Clearly we have a contradiction, since $\cal{R}'$ is a non-empty set. Therefore, if $\bm{a}\in\cal{R}'$ we must have at least one $a_n$ that satisfies Eq.~\eqref{eq:aux_stab_first}. For simplicity of exposition, let us suppose that the $a_n$ that satisfies this restriction is from the first traffic class ($n=1$). The next step is to show that as in the set $\bigcup_{\nu\in\cal{P}}\cal{S}_\nu$, the set $\cal{R}'$ also requires that we have at least one $a_n$, aside from $a_1$, that satisfies
\begin{equation*}
	\dfrac{\phi_n}{P_n^\delta}\,\dfrac{a_n}{1-a_n} <
    \dfrac{1-\phi_1\,\lambda_1\,a_1}{P_1^\delta\,\lambda_1\,a_1+
    \sum_{k=2}^N P_k^\delta\,\lambda_k}.
\end{equation*}
We can also prove this by contradiction and then, for simplicity, suppose that $a_2$ is the one that satisfies this restriction. We repeat this procedure until we reach all the $N$ traffic classes. Let us show the $j$-th step for completeness. Suppose that for all $n \in \{j,j+1,\dots,N\}$,
\begin{equation*}
	\dfrac{\phi_n}{P_n^\delta}\,\dfrac{a_n}{1-a_n} >
    \dfrac{1-\sum_{k=1}^{j-1} \phi_k\,\lambda_k\,a_k}
    {\sum_{k=1}^{j-1} P_k^\delta\,\lambda_k\,a_k+
    \sum_{k=j}^N P_k^\delta\,\lambda_k}.
\end{equation*}
If $\bm{a}\in\cal{R}'$, $\ell\in\cal{N}$, then by \eqref{eq:set_identity_1} and the above inequality we have that
\begin{align*}
	\dfrac{\phi_\ell}{P_\ell^\delta}\left(\dfrac{D_\ell}{D_\ell-1}\,
    \dfrac{1}{1-a_\ell}-1\right)
    &= \dfrac{\phi_n}{P_n^\delta}\left(\dfrac{D_n}{D_n-1}\,
    \dfrac{1}{1-a_n}-1\right)\\
    &> \dfrac{\phi_n}{P_n^\delta}\,\dfrac{a_n}{1-a_n}\\
    &> \dfrac{1-\sum_{k=1}^{j-1} \phi_k\,\lambda_k\,a_k}
    {\sum_{k=1}^{j-1} P_k^\delta\,\lambda_k\,a_k+
    \sum_{k=j}^N P_k^\delta\,\lambda_k}.
\end{align*}
Again, we use \eqref{eq:set_identity_1} and the above inequalities to write that
\begin{align*}
	1 &= \sum_{\ell=1}^{j-1} P_\ell^\delta\,\lambda_\ell\,a_\ell\dfrac{\phi_\ell}
    {P_\ell^\delta}\left(\dfrac{D_\ell}{D_\ell-1}\,\dfrac{1}{1-a_\ell} - 1\right) +
	\sum_{\ell=1}^{j-1} \phi_\ell\,\lambda_\ell\,a_\ell +
    \sum_{n=j}^N P_n^\delta\,\lambda_n\,\dfrac{D_n}{D_n-1}\,
    \dfrac{\phi_n}{P_n^\delta}\,\dfrac{a_n}{1-a_n} \\
    &> \left( \dfrac{1-\sum_{k=1}^{j-1} \phi_k\,\lambda_k\,a_k}
    {\sum_{k=1}^{j-1} P_k^\delta\,\lambda_k\,a_k+
    \sum_{k=j}^N P_k^\delta\,\lambda_k} \right) \left(\sum_{\ell=1}^{j-1}
    P_\ell^\delta\,\lambda_\ell\,a_\ell + \sum_{n=j}^N P_n^\delta\,\lambda_n\,
	\right) + \sum_{\ell=1}^{j-1} \phi_\ell\,\lambda_\ell\,a_\ell \\
    &= 1.
\end{align*}
As expected, we have a contradiction. Then, we must have at least one $a_n$, $n\in{\{j,j+1,\dots,N\}}$ such that
\begin{equation*}
	\dfrac{\phi_n}{P_n^\delta}\,\dfrac{a_n}{1-a_n} <
    \dfrac{1-\sum_{k=1}^{j-1} \phi_k\,\lambda_k\,a_k}
    {\sum_{k=1}^{j-1} P_k^\delta\,\lambda_k\,a_k+
    \sum_{k=j}^N P_k^\delta\,\lambda_k}.
\end{equation*}
We choose $a_j$ to satisfy the restriction. It is possible to do this for $a_1, a_2, \dots, a_N$. For simplicity of exposition, we showed the procedure in the order $a_1,a_2,\dots,a_N$, however it is easy to see that it can be done for all possible permutations. Therefore, $\cal{R}' \subset \bigcup_{\nu\in\cal{P}} \cal{S}_\nu$. However, since we included all possible delays in Eq.~\eqref{eq:set_identity_1}, we must also have that $ \bigcup_{\nu\in\cal{P}} \cal{S}_\nu \subset \cal{R}'$. Therefore, $\cal{R} = \cal{R}' = \bigcup_{\nu\in\cal{P}} \cal{S}_\nu$.
\end{proof}

\bibliographystyle{IEEEtran}
\bibliography{ref}

\end{document}

%% file: Plots/delay_1user.tex
\begin{tikzpicture}[scale=0.7]

\begin{axis}
[
  title={},
  legend style={at={(0.05,0.95)}, anchor=north west},
  xlabel={Arrival Rate $a$},
  ylabel={Mean Delay $D$ [slots]}, ylabel style={rotate=0},
  xmin = 0, 	xmax = 0.7,
  ymin = 0, 	ymax = 30,
  grid = both,
  scale only axis,
]
	\legend{$\phi\,\lambda = 2.0$, $\phi\,\lambda = 1.0$, $\phi\,\lambda = 0.5$}
    \addplot[red!80!black, domain=0:0.333, samples=200, thick]
    	{(1-x)/(1-(1+2.0)*x)};
    \addplot[green!60!black, domain=0:0.499, samples=200, thick]
    	{(1-x)/(1-(1+1.0)*x)};
    \addplot[blue!80!black, domain=0:0.666, samples=200, thick]
    	{(1-x)/(1-(1+0.5)*x)};
	\addplot[color=red!40!black, only marks, mark=+,
    	mark options={scale=1.5}] table {
		0.0100	1.0213
		0.0448	1.0901
		0.0796	1.2140
		0.1144	1.3459
        0.1493	1.5202
        0.1841	1.8112
        0.2189	2.3015
        0.2537	3.1497
        0.2885	5.2473
        0.3233	23.3754
		};
	\addplot[color=green!40!black, only marks, mark=+,
    	mark options={scale=1.5}] table {
        0.0100	1.0109
        0.0633	1.0870
        0.1167	1.1584
        0.1700	1.2466
        0.2233	1.4012
        0.2767	1.6125
        0.3300	1.9778
        0.3833	2.6350
        0.4367	4.4295
        0.4900	24.9313
		};
	\addplot[color=blue!40!black, only marks, mark=+,
    	mark options={scale=1.5}] table {
        0.0100	1.0269
        0.0819	1.0374
        0.1537	1.1070
        0.2256	1.1795
        0.2974	1.2631
        0.3693	1.4302
        0.4411	1.6675
        0.5130	2.0945
        0.5848	3.3078
        0.6567	22.1393
		};
    \addplot[blue!80!black, dashed, domain=0:0.6] table
    [
		x expr = \thisrow{a},
    	y expr = \thisrow{0.5}
    ] {./Data/delay_stamatiou.dat};
    \addplot[green!60!black, dashed, domain=0:0.6] table
    [
		x expr = \thisrow{a},
    	y expr = \thisrow{1.0}
    ] {./Data/delay_stamatiou.dat};
    \addplot[red!80!black, dashed, domain=0:0.6] table
    [
		x expr = \thisrow{a},
    	y expr = \thisrow{2.0}
    ] {./Data/delay_stamatiou.dat};
\end{axis}

\end{tikzpicture}

%% file: Plots/Opt_Delay.tex
\begin{tikzpicture}[scale=0.7]

\begin{axis}
[
  title={},
  legend style={at={(0.95,0.95)}, anchor=north east},
  legend cell align=left,
  xlabel={Coefficient $c_2$},
  ylabel={Mean Delay [slots]}, ylabel style={rotate=0},
  xmin = 0.1, 	xmax = 1.0,
  ymin = 1, 	ymax = 6,
  grid = both,
  scale only axis,
]
	\legend{$D_1$, $D_2$, Weighted delay}
    \addplot[red, thick] table
    [
		x expr = \thisrow{c2},
    	y expr = \thisrow{D1}
    ] {./Data/Opt_Delay.dat};
    \addplot[black, dashed, thick] table
    [
		x expr = \thisrow{c2},
    	y expr = \thisrow{D2}
    ] {./Data/Opt_Delay.dat};
    \addplot[blue, dotted, thick] table
    [
		x expr = \thisrow{c2},
    	y expr = \thisrow{Davg}
    ] {./Data/Opt_Delay.dat};
\end{axis}

\end{tikzpicture}

%% file: Plots/Opt_Delay_a.tex
\begin{tikzpicture}[scale=0.7]

\begin{axis}
[
  title={},
  legend style={at={(0.05,0.95)}, anchor=north west},
  xlabel={Arrival Rate $a_2$},
  ylabel={Mean Delay [slots]}, ylabel style={rotate=0},
  xmin = 0, 	xmax = 0.7,
  ymin = 1, 	ymax = 5,
  grid = both,
  scale only axis,
]
	\legend{$c_2 = 0.2$, $c_2 = 1.0$, $c_2 = 5.0$}
    \addplot[blue!80!black, thick] table
    [
		x expr = \thisrow{a2},
    	y expr = \thisrow{D2_.2}
    ] {./Data/Opt_Delay_a.dat};
    \addplot[green!60!black, thick] table
    [
		x expr = \thisrow{a2},
    	y expr = \thisrow{D2_1}
    ] {./Data/Opt_Delay_a.dat};
    \addplot[red!80!black, thick] table
    [
		x expr = \thisrow{a2},
    	y expr = \thisrow{D2_5}
    ] {./Data/Opt_Delay_a.dat};
        \addplot[blue!80!black, thick, dashed] table
    [
		x expr = \thisrow{a2},
    	y expr = \thisrow{D1_.2}
    ] {./Data/Opt_Delay_a.dat};
    \addplot[green!60!black, thick, dashed] table
    [
		x expr = \thisrow{a2},
    	y expr = \thisrow{D1_1}
    ] {./Data/Opt_Delay_a.dat};
    \addplot[red!80!black, thick, dashed] table
    [
		x expr = \thisrow{a2},
    	y expr = \thisrow{D1_5}
    ] {./Data/Opt_Delay_a.dat};
\end{axis}

\end{tikzpicture}

%% file: Plots/Opt_RatioPower.tex
\begin{tikzpicture}[scale=0.7]

\begin{axis}
[
  title={},
  legend style={at={(0.97,0.97)}, anchor=north east},
  xlabel={Arrival Rate $a_2$},
  ylabel={Transmit Power Ratios $P_1/P_2$}, ylabel style={rotate=0},
  xmin = 0, 	xmax = 0.7,
  ymin = 0, 	ymax = 5,
  grid = both,
  scale only axis,
]
	\legend{$c_2 = 0.2$, $c_2 = 1.0$, $c_2 = 5.0$}
    \addplot[blue!80!black, thick] table
    [
		x expr = \thisrow{a2},
    	y expr = \thisrow{0.2}
    ] {./Data/Opt_RatioPower.dat};
    \addplot[green!60!black, thick] table
    [
		x expr = \thisrow{a2},
    	y expr = \thisrow{1}
    ] {./Data/Opt_RatioPower.dat};
    \addplot[red!80!black, thick] table
    [
		x expr = \thisrow{a2},
    	y expr = \thisrow{5}
    ] {./Data/Opt_RatioPower.dat};
\end{axis}

\end{tikzpicture}

%% file: Plots/max_a1.tex
\begin{tikzpicture}[scale=0.7]

\begin{axis}
[
  title={},
  legend style={at={(0.95,0.05)}, anchor=south east},
  xlabel={Maximum Mean Delay $D_1^*$ [slots]},
  ylabel={Maximum Arrival Rate $a_1$}, ylabel style={rotate=0},
  xmin = 0, 	xmax = 10,
  ymin = 0, 	ymax = .5,
  grid = both,
  scale only axis,
]
	\legend{$\Psi_2 = 0.0$, $\Psi_2 = 0.3$, $\Psi_2 = 0.5$, $\Psi_2 = 0.7$}
    \addplot[yellow!50!black, thick] table
	[
		x expr = \thisrow{D1},
    	y expr = \thisrow{0}
    ] {./Data/max_a1.dat};
    \addplot[blue!80!black, thick] table
    [
		x expr = \thisrow{D1},
    	y expr = \thisrow{0.3}
    ] {./Data/max_a1.dat};
    \addplot[green!80!black, thick] table
    [
		x expr = \thisrow{D1},
    	y expr = \thisrow{0.5}
    ] {./Data/max_a1.dat};
    \addplot[red!80!black, thick] table
	[
		x expr = \thisrow{D1},
    	y expr = \thisrow{0.7}
    ] {./Data/max_a1.dat};
\end{axis}

\end{tikzpicture}

%% file: Plots/P1P2.tex
\begin{tikzpicture}[scale=0.7]

\begin{axis}
[
  title={},
  legend style={at={(0.05,0.95)}, anchor=north west},
  xlabel={Maximum Mean Delay $D_1^*$ [slots]},
  ylabel={Transmit Power Ratio $\frac{\phi_1}{\phi_2}\frac{P_2^\delta}{P_1^\delta}$}, ylabel style={rotate=0},
  xmin = 0, 	xmax = 10,
  ymin = 0, 	ymax = 3,
  grid = both,
  scale only axis,
]
	\legend{$\Psi_2 = 0.7$, $\Psi_2 = 0.5$, $\Psi_2 = 0.3$}
    \addplot[red!80!black, thick] table
	[
		x expr = \thisrow{D1},
    	y expr = 1/\thisrow{0.7}
    ] {./Data/P1P2.dat};
    \addplot[green!80!black, thick] table
    [
		x expr = \thisrow{D1},
    	y expr = 1/\thisrow{0.5}
    ] {./Data/P1P2.dat};
    \addplot[blue!80!black, thick] table
    [
		x expr = \thisrow{D1},
    	y expr = 1/\thisrow{0.3}
    ] {./Data/P1P2.dat};
\end{axis}

\end{tikzpicture}